\documentclass[journal,twocolumn]{IEEEtran}

\usepackage{color}
\usepackage{graphicx}
\usepackage{amsfonts}
\usepackage{amsmath}
\usepackage{graphicx}
\usepackage{cite}
\usepackage{epsfig}
\usepackage{latexsym}
\usepackage{subfigure}
\usepackage{epstopdf}
\usepackage{color}
\usepackage{verbatim}
\usepackage{units}
\usepackage{amsthm}
\usepackage[longend,ruled]{algorithm2e}

\newcommand{\mat}{\boldsymbol}

\newcommand{\AND}{\mathbf{AND}}

\newtheorem{lemma}{Lemma}
\newtheorem{Corollary}{Corollary}

\begin{document}
\title{On Optimizing Power Allocation For Reliable Communication over Fading Channels with Uninformed Transmitter}

\author{M.~Majid~Butt,~\IEEEmembership{Senior~Member,~IEEE,}~Eduard~A.~Jorswieck,~\IEEEmembership{Senior~Member,~IEEE}\\~and~Nicola~Marchetti,~\IEEEmembership{Senior~Member,~IEEE}

\thanks{The material in this paper has been presented in part in ICC 2017, Paris, France \cite{Majid:ICC2017}.}
\thanks{M. Majid Butt and Nicola~Marchetti are with CONNECT center for future networks, Trinity College Dublin, Dublin 2, Ireland. Email:\{majid.butt, nicola.marchetti\}@tcd.ie.}
\thanks{Eduard A. Jorswieck is with Department of Electrical Engineering and Information Technology, TU Dresden, Germany. Email: eduard.jorswieck@tu-dresden.de.}
\thanks{This publication has emanated from research conducted with the financial support of Science Foundation Ireland (SFI) and is co-funded under the European Regional Development Fund under Grant Number 13/RC/2077.}

\thanks{The work of E. Jorswieck was supported by the German Research Foundation, Deutsche Forschungsgemeinschaft,
Collaborative Research Center 912, through the Highly Adaptive Energy- Efficient Computing.
}
}

\maketitle
\begin{abstract}
We investigate energy efficient packet scheduling and power allocation problem for the services which require reliable communication to guarantee a certain quality of experience (QoE). We establish links between average transmit power and reliability of data transfer, which depends on both average amount of data transfer and short term rate guarantees. We consider a slow-fading point-to-point channel without channel state information at the transmitter side (CSIT). In the absence of CSIT, the slow fading channel has an outage probability associated with every transmit power. As a function of data loss tolerance parameters, and minimum rate and peak power constraints, we formulate an optimization problem that adapts rate and power to minimize the average transmit power for the user equipment (UE). Then, a relaxed optimization problem is formulated where transmission rate is assumed to be fixed for each packet transmission. We use Markov chain to model constraints of the optimization problem. The corresponding problem is not convex for both of the formulated problems, therefore a stochastic optimization technique, namely the simulated annealing algorithm, is used to solve them. The numerical results quantify the effect of various system parameters on average transmit power and show significant energy savings when the service has less stringent requirements on timely and reliable communication.
\end{abstract}

\section{Introduction}
Internet of things (IoT) is one of the use cases of 5G wireless communications to serve the heterogeneous services. Services like smart city, smart buildings and smart transportation systems depend heavily on efficient information processing and reliable communication techniques. The use of thousands of smart and tiny sensors to communicate regular measurements, e.g., temperature, traffic volume, etc., makes it extremely important to look at the energy efficiency aspect of the problem. Achieving ultra reliability and low latency communication (URLLC) at low energy for the emerging applications in 5th generation (5G) of wireless communication is considered very challenging \cite{She_commag:2017}. Due to short packet size in IoT and machine type communication, finite block-length channel codes, novel diversity techniques, packet dropping mechanisms and control plane communication strategies are considered to enable URLLC \cite{Trillingsgaard_TCOM:2017,She_TWC:2018,Gursoy_Eurasip:2013,Xu_TCOM:2016}.

In 5G networks, context aware scheduling is believed to play key role in smart use of resources \cite{Pahalawatta:TVT2007} and the requirements on reliability and latency are dictated by the nature of the application. More specifically, IoT applications have extremely heterogenous requirements in terms of (average or deadline) latency, reliability and frequency of packet transmissions, and require quality of service (QoS) aware resource allocation mechanisms \cite{Li:TII_2014}. Depending on the application's context, it may not be necessary to receive every packet correctly at the receiver side to avoid experiencing a serious degradation in quality of experience (QoE). For instance, ITU recommendation ITU-T G.1080 (12/2008) specifies a set of requirements for picture/audio that define the quality impairments in addition to average packet loss rates \cite{ITU_standard}. If some packets are lost, the application may tolerate the loss without requiring retransmissions of the lost packets. The application loss tolerance without degrading quality can effectively be exploited to reduce average energy consumption of the devices.

We investigate energy efficient power allocation for the wireless systems with data loss constraints. The packet loss constraints are defined in terms of average packet loss and the maximum number of successively lost packets. The reliability aspect of the communication systems is conventionally handled at upper layers of communication using error correction codes and/or hybrid automatic repeat request (HARQ). Feedback based link adaptation applied in HARQ is dictated by the latency constraints of the application \cite{Choi_TVT:2013}. Our approach is different from the HARQ scheme because the simple device nodes do not possess a data buffer, which makes implementing HARQ systems impossible. Instead, we assume that the applications's QoE does not require every packet to be received successfully, i.e., loss of successive packets can be tolerated, but it must be bounded and parameterized. Video streaming, video conferencing, disaster management systems and interactive gaming are examples of such applications.

In literature, some earlier works have addressed similar problems in different settings and contexts (more at network level). In video streaming applications, it is important to select source coding parameters for various encoded representations of the same content in order to minimize the consumption power while maintaining a high quality of experience for the users \cite{Li:TVT_2017}. Packet scheduling for media streaming with network coding has been studied in \cite{Sheikh:TMM_2014} with the objective to improve the perceived media quality. Wu \emph{et al.} address the problem of bursty packet loss over internet in \cite{Wu:2016}. The authors propose a transmission scheme that trades delay to reduce the distortion in transmitted data. Similar works in \cite{Wu_JSAC:2016,Wu_JSAC:2017} address delay-quality tradeoff in video transmission over communication links. In \cite{Nasralla:2014}, the authors evaluate the subjective and objective performance of video traffic for bursty loss patterns. Reference \cite{Zou:2013} considers real-time packet forwarding over wireless multi-hop networks with lossy and bursty links. The objective is to maximize the probability that individual packets reach their destination before a hard delay deadline. In a similar study, the authors in \cite{Aditya_TWC:2010} investigate a scenario where multimedia packets are considered lost if they arrive after their associated deadlines. Lost packets degrade the perceived quality at the receiver, which is quantified in terms of the "distortion cost" associated with each packet. The goal of the work in \cite{Aditya_TWC:2010} is to design a scheduler which minimizes the aggregate distortion cost over all receivers.

The energy efficiency aspect of the problem has been discussed in many works. Energy can be saved by relaxing various QoS constraints for data transmission. Delay and loss tolerance are two key dimensions to exploit for reducing transmit power. Various works in literature deal with exploiting delay tolerance to optimize transmit power in time varying wireless channels, e.g. \cite{Rajan:IT2004,Lee:TWC_2009,Lee:IT2013,Berry:TIT_2013}. If the latency requirements for the data permit, the transmission can be delayed and the effect of the random nature of fading wireless channels can be minimized by opportunistic scheduling schemes.
The energy aspect of the problem has been addressed in \cite{Neely2009} where the authors investigate intentional packet dropping mechanisms for delay limited systems to minimize energy cost over fading links.

Most of the works in literature characterize performance of the wireless network for average packet loss. In addition to average packet loss, bursty data loss is an important phenomenon which needs to be defined, characterized and analyzed. Some works analyze system for bursty traffic, e.g., the effect of access router buffer size on packet loss rate is studied in \cite{Sequeira:2013} when bursty traffic is present. However, assumption of bursty traffic is different from the notion of bursty data loss. An analytical framework to dimension the packet loss burstiness over generic wireless channels is considered in  \cite{Fanqqin:2013} and a new metric to characterize the packet loss burstiness is proposed. However, these works do not characterize the effect of average and bursty packet loss on the consumed energy at link level. Some recent studies in \cite{Majid:ICC2017,majid_TWC:13, majid:sys2016} characterize the effect of packet loss burstiness on average system energy for a multiuser wireless communication system where the transmit channel state information (CSIT) is fully available or erroneous.

In this work, no CSIT is assumed to be available, which poses new challenges for communication and scheduler design. When CSIT is not available for slow fading channels, channel state dependent power control cannot be applied and outage free communication cannot be guaranteed. For the no-CSIT case, we characterize the average power consumption of the point-to-point wireless channel for various average and bursty packet drop parameters, as well as the outage probability that the application can tolerate loss of a full sequence of packets (successively).

The main contributions of the work are summarized as follows:
\begin{itemize}
  \item We model and formulate the power optimization problem for a point to point system using a Markov chain. The problem constraints involve various parameters that help characterizing QoE for a particular application, including average and successive packet loss bounds, as well as minimum packet size and long term average rate guarantees. We show that the formulated optimization problem is combinatorial and no closed form solution exists.
  \item We propose a solution of the optimization problem based on a low complexity stochastic optimization algorithm, namely Simulated Annealing (SA). The algorithm is based on randomization of input parameters. We numerically evaluate the performance of the proposed solution and verify that the algorithm produces results that are very close to the analytical solution for a special case of the problem.
  \item To reduce the complexity of the problem, we propose a fixed-rate adaptive power transmission scheme. The fixed rate transmission scheme inherits all the constraints of the original optimization problem, but the transmitted rate is the same for every transmission. This helps in reducing computational complexity for the problem.
  \item Simulation results show that our power allocation scheme exploits packet loss tolerance of the application to save considerable amount of energy; and thereby significantly improves the energy efficiency of the network as compared to lossless application case.
\end{itemize}

The rest of the paper is organized as follows. The system model for the work is introduced in Section \ref{sect:system_model} and state space description of the proposed scheme is discussed in Section \ref{sect: state space}. We formulate the optimization problem to minimize average transmit power in Section \ref{sect:optimization}. Then, we discuss a modified optimization problem in Section \ref{sec:fixed_rate} where all the transmissions are of fixed rate. We discuss solution of both of the optimization problems using SA algorithm in Section \ref{sect:stochastic}. The performance of the proposed framework is numerically evaluated in Section \ref{sect:results} and Section \ref{sect:conclusions} summarizes the main results of the paper.

\section{System Model}
\label{sect:system_model}
We consider a point-to-point system such that the transmitter user equipment (UE) has a single packet to transmit in each time slot. The packets are assumed to be variable in size, measured in bits/s/Hz. This is achieved by rate adaptation at physical layer using well known adaptive modulation and coding techniques. Time is slotted and the UE experiences quasi-static independent and identically distributed (i.i.d) block flat-fading such that the fading channel remains constant for the duration of a block, but varies from block to block whereas duration of the block equals one time slot.

We assume no CSIT, but the transmitter is aware of the fading channel distribution. Depending on the scheduling state $i$ (explained later in Section \ref{sect: state space}), the UE transmits with a fixed power $P_i\leq P_{m}$ to transmit a packet with size $R_i$ bits/s/Hz, and waits for the feedback. $P_m$ is the peak transmit power constraint for the transmitter. For convenience, the distance between the transmitter and the receiver is assumed to be normalized.

For a transmit power $P_i$, and channel fading coefficient $h$, the outage probability for the failed transmission (channel outage) is denoted by $\epsilon_i$ such that,
\begin{equation}\label{eqn:outage}
  \epsilon_i=\Pr\left[\log_2\left(1+\frac{P_i|h|^2}{N_0}\right)<R_i\right]
\end{equation}
where $N_0$ is additive white Gaussian noise power.

If the transmitted packet is received at the receiver correctly, the receiver sends back a positive acknowledgement (ACK) message to the UE. If it is not decoded at the receiver, a negative acknowledgement (NAK) is fed-back to the UE. The feedback is assumed to be perfect without error. Note that a power and/or rate adaptation based on the feedback can be applied even without CSIT.
Feedback based power allocation belongs to Restless Multi-armed Bandit Processes where the states of the UE in the system stochastically evolve based on the current state and the action taken. The UE receives a reward depending on its state and action. The next action depends on the reward received and the resulting new state. In this work, we investigate the effect of feedback based sequential decisions in terms of UE consumed average power.

\subsection{Problem Statement}
A single packet arrives at the transmit buffer of the UE in every time slot. The UE's data buffer has no capacity to store more than one packet.\footnote{Note that the buffer capacity is given by the largest rate $R_{max}$ bits/s/Hz, that can be transmitted in any state.} This is a typical scenario for a wireless sensor network application where data measurements arrive constantly after regular fixed time intervals. The UE is battery powered, which needs to be replaced after regular intervals. It is therefore, important to save transmit energy as much as possible. Depending on the application, the UE has two constraints on the reliability of data packet transfer \cite{majid_TWC:13, majid:sys2016}:
\begin{enumerate}
  \item \emph{Average packet drop/loss rate} $\gamma$ is the parameter that constraints the average number of packets dropped/lost,
      \begin{equation}\label{eqn:avgdrop}
        \gamma=\lim_{t\to \infty} \frac{\mbox{Packets dropped}}{\mbox{Packets transmitted}}
      \end{equation}
  \item Maximum number of packets dropped successively. This is called \emph{bursty packet drop constraint}. The parameter $N$ denotes the \emph{maximum number of packets allowed to be dropped successively without degrading QoE below a certain level}. Mathematically, the distance $r(q,q-1)$ between $q^{th}$ and $q^{th}-1$ correctly received packets measured in terms of number of successively lost packets is constrained by parameter $N$, i.e.,
      \begin{equation}
        r(q,q-1)\leq N.
        \label{eqn:successive}
      \end{equation}
\end{enumerate}
Due to transmit power constraint, it is not possible to provide the guarantee in (\ref{eqn:successive}) with probability one. Given at least $N$ packets have been lost successively by time instant $t-1$, we define a parameter $\epsilon_{out}$ at an instant $t$ by the probability that the $N+1-th$ packet is lost, i.e.,
      \begin{equation}
        \epsilon_{out}=\Pr\Big(r_t (q,q-1) = r_{t-1}(q,q-1)+1|r_{t-1}(q,q-1)\geq N\Big)
        \label{eqn:bursty_loss}
      \end{equation}

\begin{figure}[t]
\centering
\includegraphics[width=3.0in]{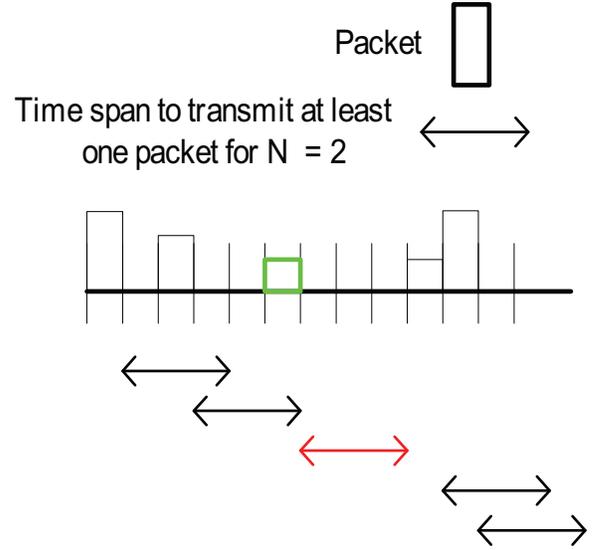}
\caption{Schematic diagram of the system with $N=2$. The transmitted packets can be of variable size as shown in the diagram. The time span for the successive packet loss constraint has been drawn. The red period shows the instance when violation of successive packet loss occurs. The transmission of green packet in third time slot shows that constraint $N=2$ was met.}
\label{fig:schematic}
\end{figure}

Fig. \ref{fig:schematic} shows the schematic diagram for the system. The successive packet drop parameter $N$ equals 2 in the diagram. Whenever a packet is transmitted successfully, it is permitted to drop 2 packets in the next 2 time slots. If a packet is not transmitted successively in $3^{\rm rd}$ successive time slot, it counts as an outage $\epsilon_{out}$. The span of 3 successive packet drops have been depicted as red in the schematic diagram. Note that the transmitted packets have variable size in bits/s/Hz, which is constrained by minimum and maximum rate $R_{min}$ and $R_{max}$, respectively. If a packet is transmitted immediately after the transmission of a packet in previous time slot, its size (rate) is more as compared to the packets transmitted when the packets have been lost already. We come back to rate adaptation and optimization later in the next section.

All of these parameters described in this section contribute to the QoE for the application. The average packet drop rate is commonly used to characterize a wireless network and bounds the QoE for the application. However, bursty packet loss in the applications like smart monitoring sensors can degrade the performance enormously due to absence of contiguous data measurements. On the other hand, the UE can exploit the parameters $\gamma$ and $N$ to optimize average energy consumption if the application is more loss tolerant. If the application is loss tolerant and packet size is fixed, it is advantageous to transmit with a small power if a packet has just been received successfully in the last time slot because the impact of packet loss due to outage is not so severe on cumulative QoE. The consideration of bursty (successive) packet loss poses a new challenge in system modeling as the number of packets lost in previous time slots affect the power allocation decision at time slot $t$.
Clearly, there is a trade-off between transmitting a packet at time $t$ with small power based on the success of transmission in time slots $[t-1, t-2,\dots]$, and transmitting with large power to lower the risk of outage. This trade-off determines the power allocation policy. Let us illustrate the impact of ACKs and NAKs on the tightness of the constraints in the following:

If the permitted average packet loss rate $\gamma$ is very high but $N$ is small, i.e., it is not permitted to lose more than $N$ packets successively without degrading QoE, the effective average packet drop rate becomes much lower than the permitted $\gamma$ in this case. It may work to transmit with small power due to large $\gamma$, but parameter $N$ does not allow it. Due to successive packet drop constraint $N$, transmission of a packet in a time slot $t$ may not be as critical as in any other time slot with $t'\ne t$. If a packet was transmitted successfully in a time slot $t-1$, it implies that transmitting a packet with a lower power is not as risky in time slot $t$. However, when the number of successively lost packets approaches $N$, power allocation needs to be increased proportionally to avoid/minimise the event of missing $N$ packets successively, which may cause loss of important information for wireless sensor networks.

A similar justification can be provided for rate adaptation for transmission in various states. If a transmission is made in the beginning, packet rate can be chosen a bit higher as risk involved due to dropping of a packet is not that great. When more successive packets are dropped, the rate must be decreased to increase the probability of success as depicted in Fig. \ref{fig:schematic}. The rate is lower bounded by $R_{min}$, a system parameter, while the upper bound is obtained from the Shanon capacity with the peak power constraint $P_m$ such that,
\begin{equation}\label{eqn:maxrate}
  R_{max}=\log_2\left(1+\frac{P_m|h|^2}{N_0}\right)
\end{equation}
With every unsuccessful packet transmission, the response of the transmitter is to reduce the rate to increase the success probability, though it is not straight forward to see how this adaptation needs to be applied. The objective is to reduce the average transmit power, therefore rate and power adaptation with every unsuccessful transmission should be optimized in a way that QoE in terms of successful data transfer according to the parameters provided should be met and the transmit power is not wasted unnecessarily.

\section{State Space Description}
\label{sect: state space}
To model the problem, we take history of the packet transmissions in the last $N$ time slots into account. If a NAK is received in time slot $t-1$, it needs to be determined whether transmission in time slot $t-2$ was an ACK or NAK. To capture the time evolution of the packet transmissions in successive time slots, we model the problem using a Markov chain where the next state only depends on the current state and is independent of the history. In a Markov chain, Markov state $i$ is defined by the number of packets lost successively at the transmit time $t$. If a packet was transmitted successfully in time slot $t-1$, the current state $i=0$. If two successive packets are lost in time slots $t-1$ and $t-2$, $i=2$. The maximum number of Markov states is determined by parameter $N$, i.e., the bursty packet drop constraint. In the following, we explain how Markov chain process can effectively be used to model and formulate power allocation optimization problem for average and bursty packet drop constraints.

To explain the state transition mechanism, let us examine the power allocation policy first.
At the beginning of the Markov chain process, a packet is transmitted with power $P_0$ and rate $R_0$ in a time slot $t$ with initial Markov state $i=0$. The channel has an outage probability of $\epsilon_i$ (defined in (\ref{eqn:outage})).
If the received feedback is ACK, the process moves back to state $0$, otherwise moves to state $1$. The lost packet is dropped permanently as UE has no buffer. In state $i=1$, the new arriving packet is transmitted with power $P_1$ and rate $R_1\geq R_{min}$. Thus, power allocation in state $i$ is a function of outage probability $\epsilon_i$ and the rate $R_i$,
\begin{equation}
P_i=f(\epsilon_i,R_i).
\end{equation}
If the packet is transmitted successfully, the next state is zero, 2 otherwise. Similarly, the Markov chain makes a transition to either state $i+1$ or state zero corresponding to the event of unsuccessful or successful transmission, respectively.
When $i=N$ (termination state) and a packet is not transmitted successfully, this defines the outage event for successive packet loss. This is modeled by self state transition probability $\alpha_{NN}$ of staying in state $S_N$ such that,
\begin{eqnarray}
  \alpha_{NN}=\epsilon_N
  =\Pr(S_{t+1}=N|S_t=N).
\end{eqnarray}
$P_N$ is chosen such that $\alpha_{NN}\leq \epsilon_{out}$ where $\epsilon_{out}$ is a system parameter defined in (\ref{eqn:bursty_loss}). If a packet is lost in state $N$, we want Markov process to stay in state $N$ for the next time slot to minimize further degradation in QoE as rate and power levels in state $N$ are designed to maximize the possibility of successful transmission.

The state transitions from state $i$ to $j$ occur with a state transition probability $\alpha_{ij}$.
It is a function of parameters $\gamma, N$ and channel distribution. For every transmit power $P_i$, there is an associated state transition probability $\alpha_{ij}$.

Formally, the state transition probability $\alpha_{ij}$ from the current state $S_t=i$ to next state $S_{t+1}=j$ is defined by,
\begin{eqnarray}
\alpha_{ij} &=& {\rm Pr}(S_{t+1}=j|S_t=i)\\
&=&\begin{cases}
              1-\epsilon_i, & \mbox{if ACK Received},\forall i, j=0  \\
              \epsilon_i, & \mbox{if NAK Received},i\ne N,j=i+1,\\&0 \leq \epsilon_i \leq 1 \\
              \epsilon_N,&\mbox {if NAK Received},i= N,j=N  \\
              0, & \mbox{otherwise}
            \end{cases}
\end{eqnarray}
where $\epsilon_i$ is given by (\ref{eqn:outage}).
The resulting state diagram is shown in Fig. \ref{fig:state_dia}.
The state transition probability matrix $\mathbf{A} = [\alpha_{ij}]_{i,j=0}^N$ takes the form
\begin{equation}\label{eqn:A}
  \mathbf{A}=\left(
               \begin{array}{ccccc}
                 1-\epsilon_0 & \epsilon_0 & 0 &\dots & 0 \\
                 1-\epsilon_1  & 0& \epsilon_1 &\dots& 0 \\
                  \ddots&\ddots&\ddots&\ddots&\ddots\\
                  1-\epsilon_{N-1} & 0 & 0 &\dots& \epsilon_{N-1} \\
                 1-\epsilon_N & 0 & 0 &\dots& \epsilon_{N} \\
               \end{array}
             \right)
\end{equation}
For a time homogeneous Markov chain, the steady state probability for state $j$, $\pi_j$ is defined by
\begin{equation}
\pi_j=\sum_{i\in \mathcal{S}}\alpha_{ij}\pi_i
\end{equation}
where $\mathcal{S}$ defines the state space for the UE states.
Assuming $N_0=1$, for Rayleigh fading and state $i$, the outage probability is given by,
\begin{eqnarray}
  \epsilon_i&=& 1-\exp\Big(\frac{-(2^{R_i}-1)}{P_i}\Big)
\end{eqnarray}
After some algebraic manipulation, the required transmit power $P_i$ is calculated by,
\begin{equation}
P_i=\frac{1-2^{R_i}}{\log(1-\epsilon_i)}
\label{eqn:power}
\end{equation}

\begin{figure}
\center
\includegraphics[width=3.5in]{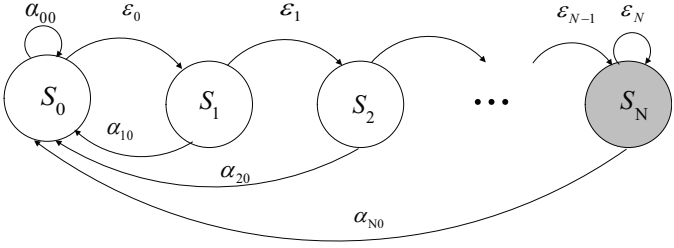}

\caption{State diagram for the Markov chain for the UE power allocation scheme.}
\label{fig:state_dia}

\end{figure}
From the transmit power for every state $i$, the average transmit power consumed is given by,
\begin{equation}\label{eqn:avg_power}
\bar{P}=\sum_{i=0}^N P_i \pi_i.
\end{equation}

\section{Optimization Problem Formulation}
\label{sect:optimization}
The optimization problem is to jointly compute a vector of power values $\mathbf{P}=[P_0,P_1,\dots P_N]$ and $\mathbf{R}=[R_0,R_1,\dots R_N]$ , which satisfies the constraints on packet dropping parameters and minimizes average system energy. The problem is mathematically formulated as,

\begin{eqnarray}
\label{eqn:optimization}
&&\min_{\mathbf{P,R}} \Bar{P}=\sum_{i=0}^N P_i\pi_i\\
s.t.
&&\begin{cases}
\mathcal{C}_1:\sum_{i=0}^N R_i\pi_i\geq R,\quad 0\leq i\leq N\\
\mathcal{C}_2:\sum_{i=0}^{N}\epsilon_{i}\pi_i \leq \gamma\\
\mathcal{C}_3:\epsilon_N \leq \epsilon_{out}\\
\mathcal{C}_4:R_{min}\leq R_i\leq R_{max}
\end{cases}
\label{eqn:constrains}
\end{eqnarray}
The constraints are explained in the following:
\begin{itemize}
  \item $\mathcal{C}_1$ is the average rate constraint, i.e., the average transmitted rate should be greater than $R$.
  \item {$\mathcal{C}_2$ is the average packet loss constraint for the target average packet loss probability $\gamma$. The left hand side term is denoted by achieved average packet loss probability $\gamma_r$. From the state space model described in Section \ref{sect: state space}, it is computed by the sum of the products of steady state and forward state transition probabilities.\footnote{State $N$ is exception where self state transition represents packet loss.}

The outage probability $\epsilon_i$ and the corresponding transmit power $P_i$ for a UE in state $i$ is computed such that the average packet dropping probability constraint $\mathcal{C}_2$ holds.
\item $\mathcal{C}_3$ is the outage constraint. For $i=N$, $\epsilon_N\leq \epsilon_{out}$ where $\epsilon_{out}$ is defined in (\ref{eqn:bursty_loss}). $\epsilon_i$ cannot be determined directly and needs to be optimized for the system parameters.
\begin{equation}
\epsilon_i=f(\gamma,N,\epsilon_{out},h_X(x),R)
\end{equation}
where $h_X(x)$ is the fading channel distribution.}
  \item In $\mathcal{C}_4$, rate $R_i$ is constrained by $R_{min} \leq R_i\leq R_{max}$. $R_{min}$ is the minimum rate that a packet is expected to provide; and depends on the application and the chosen modulation and coding schemes. If we take the example of IoT, we can define a minimum non-zero value of the rate $R_{min}$ that carries the minimum information about the sensed data.
\end{itemize}
This solution of the problem provides successive as well as average packet loss guarantees. Similar to effective capacity, which gives the delay-limited capacity depending on the buffer decay rate \cite{Wu:TWC2003}, this solution also provides a minimum statistical rate guarantee $R_{min}$ with outage $\epsilon_{out}$ over the span of $N$ successive time slots and average rate guarantee $R$ when $t \to \infty$. The packet size $R_i$ is an optimization variable to be jointly computed with $P_i$ for $\forall i$. The offline computed power allocation solution holds for online power allocation as long as the channel distribution remains the same.

The optimization problem is to jointly find rate $\mathbf{R}$ and power $\mathbf{P}$ vectors that result in minimum average power. If we choose $P_i$ too high for small states (states with low number of successive outages), the packets will more likely be transmitted too early at the expense of larger power budget without exploiting loss tolerance of the application and provide good (but unnecessary) QoE. On the other side, if $P_i$ is chosen too low in the beginning, the packets will be lost mostly and we have to transmit with much higher power to meet the \emph{forced} condition that at least one packet has to be transmitted to avoid the sequence of $N$ lost packets.

The formulated optimization problem covers the discussed practical aspects of the IoT applications. We take care of both packet reception frequency (average), packet reception order (successive packet) as well as min rate $R_{min}$ transmitted within any span of $N$ transmissions, and long term average rate $R$ transported by these packets.

\subsection{Complexity of the Programming Problem}
\label{sec:complex}
There are two main difficulties associated with programming problem (\ref{eqn:optimization}). The first challenge is to get a tractable expression for the steady state distribution $\mathbf{\pi}$. It is obtained from the eigenvector of the state transition matrix $\mathbf{A}$ whose components $\alpha_{ij}$ depend on the outage probability vector $\vec{\epsilon}$ which depends on the transmission rate and power allocation (and thus on the optimization variables $\mathbf{P}$ and $\mathbf{R}$). Even though the state transition matrix $\mathbf{A}$ has a special structure as pointed out in (\ref{eqn:A}), there is no closed form solution for the eigenvectors of this structured matrix. Therefore, the dependency of $\mathbf{\pi}$ on $\mathbf{P}, \mathbf{R}$ via $\mathbf{A}$ prohibits an analytical presentation.

The second challenge arises from the structure of the programming problem as such. In order to derive an efficient algorithm to solve (\ref{eqn:optimization}), the problem should be jointly convex in $\mathbf{R}$ and $\mathbf{P}$. Let us only consider the constraint on the average dropping rate $\sum_{i=0}^N \epsilon_i \pi_i \leq \gamma$. The left side contains the expression $\epsilon_i = 1 - \exp^{-\frac{1-2^R_i}{P_i}} = \phi(R_i, P_i)$ which is a function of $R_i$ and $P_i$. The constraint in (\ref{eqn:constrains}) requires an upper bound on the average dropping rate. Since $\phi(R_i, P_i)$ is a concave function, this leads to a non-convex constraint. Therefore, even if a closed form solution for the steady state distribution could be derived, it will not lead to a convex programming problem.
In order to gain more insights into the solution structure of the programming problem (\ref{eqn:optimization}), we consider the special case $N=1$ next.

\subsection{Special Case: $N=1$}
\label{sect:N=1_general}
 In this case, the state transition probability matrix $\mathbf{A}$ reads,
\begin{equation}
  \mathbf{A}=\left(
               \begin{array}{cc}
                 1-\epsilon_0 & \epsilon_0\\
                 1-\epsilon_1 & \epsilon_1\\
               \end{array}
             \right)
\end{equation}
Steady state transition probabilities for states $0$ and $1$ are calculated as,
\begin{eqnarray}
\pi_0 = \frac{1-\epsilon_1}{1+\epsilon_0-\epsilon_1} \\
\pi_1 = \frac{\epsilon_0}{1+\epsilon_0-\epsilon_1}.
\end{eqnarray}
Computing $\gamma_r$ for $\epsilon_1=\epsilon_{out}$ and $\pi_0$ and $\pi_1$ calculated above
\begin{equation}
\gamma_r=\frac{\epsilon_0}{1+\epsilon_0-\epsilon_{out}}.
\label{eqn:gamma}
\end{equation}
We can compute the value of $\epsilon_0$ in closed form that satisfies constraints $\mathcal{C}_2$ and $\mathcal{C}_3$ with equality. Solving (\ref{eqn:gamma}) and $\mathcal{C}_2$ in (\ref{eqn:constrains}) with equality gives,
\begin{equation}
\epsilon_0= (1-\epsilon_{out})\frac{\gamma}{1-\gamma}.
\label{eqn:close_epsilon_zero}
\end{equation}
To compute power levels $P_0$ and $P_1$ for the computed $\epsilon_0$ and $\epsilon_1$, we require rates $R_0$ and $R_1$ that meet $\mathcal{C}_1$ and $\mathcal{C}_4$ and the resulting power levels minimize $\bar{P}$ from (\ref{eqn:avg_power}). There could exist many $(R_0,R_1)$ pairs that meet $\mathcal{C}_1$ and $\mathcal{C}_4$ and computing the unique combination that minimizes $\bar{P}$ in closed form is not possible.
However, under certain assumption on $\mathcal{C}_1$ and $\mathcal{C}_4$, it is possible to compute achievable rates $R_0,R_{1}$ in closed form (but not the optimal ones). As minimum acceptable transmit rate is $R_{min}$ from $\mathcal{C}_4$ and we have $N$ as the critical state, we assume that the rate transmitted in state $N$ is $R_{min}$ to maximize the chance of successful transmission.

For $N=1$ case, using $R_1=R_{min}$, and meeting $\mathcal{C}_1$ with equality, the rate $R_0$ turns out to be,
\begin{equation}
R_0=\frac{R-R_{min}\pi_1}{\pi_0}.
\label{eqn:close_ratezero}
\end{equation}
Note that the solution is feasible only if $R_0\leq R_{max}$. If $R_0>R_{max}$, it implies $R_{min}$ is not enough to meet $\mathcal{C}_1$ and needs to be increased.

\begin{figure*}
\centering
 \subfigure[$\epsilon_{out}=0.1$]
  	{\includegraphics[width=3.1in]{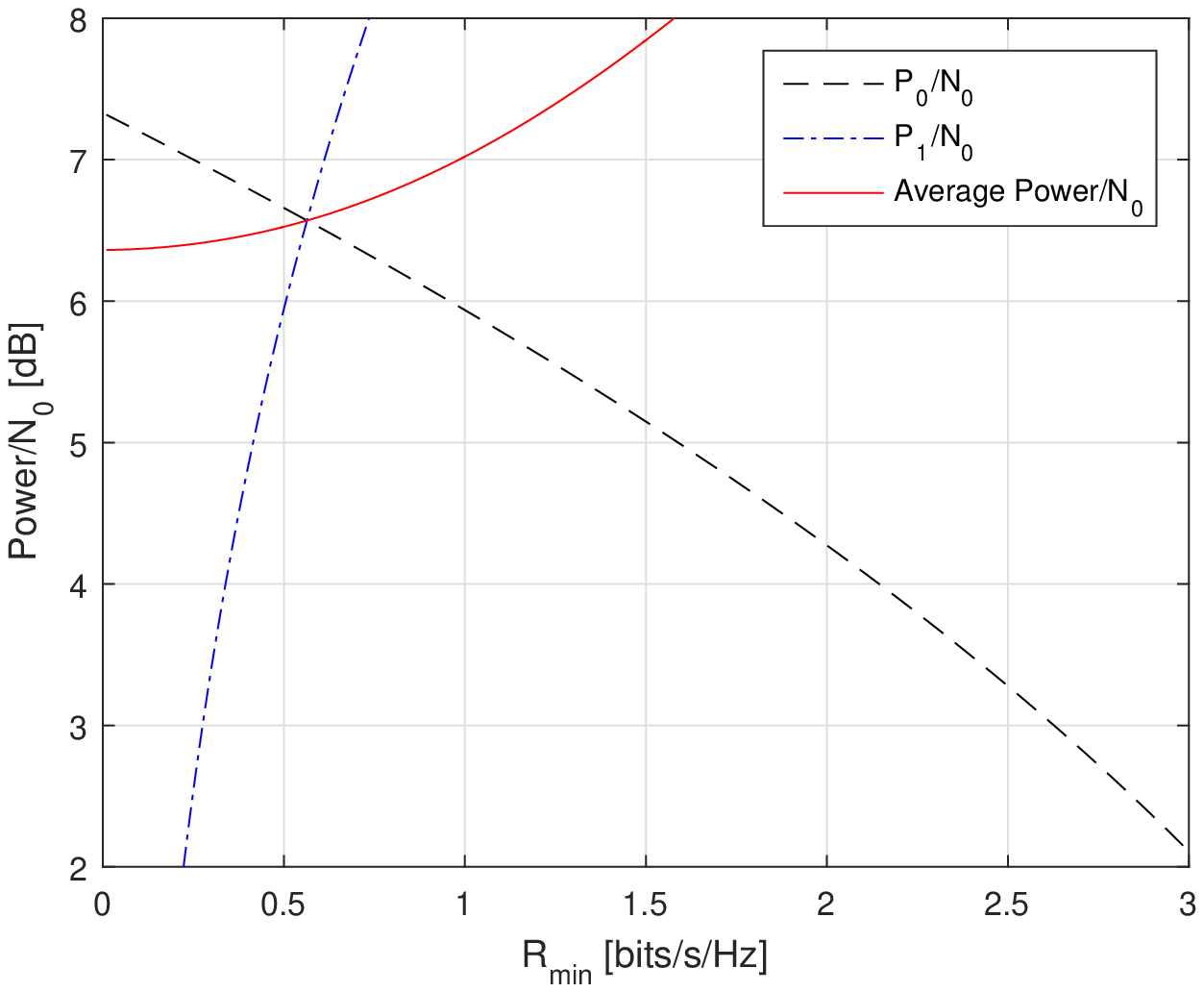}
  \label{fig:1}}
  \subfigure[$\epsilon_{out}=0.2$]
  	{\includegraphics[width=3.1in]{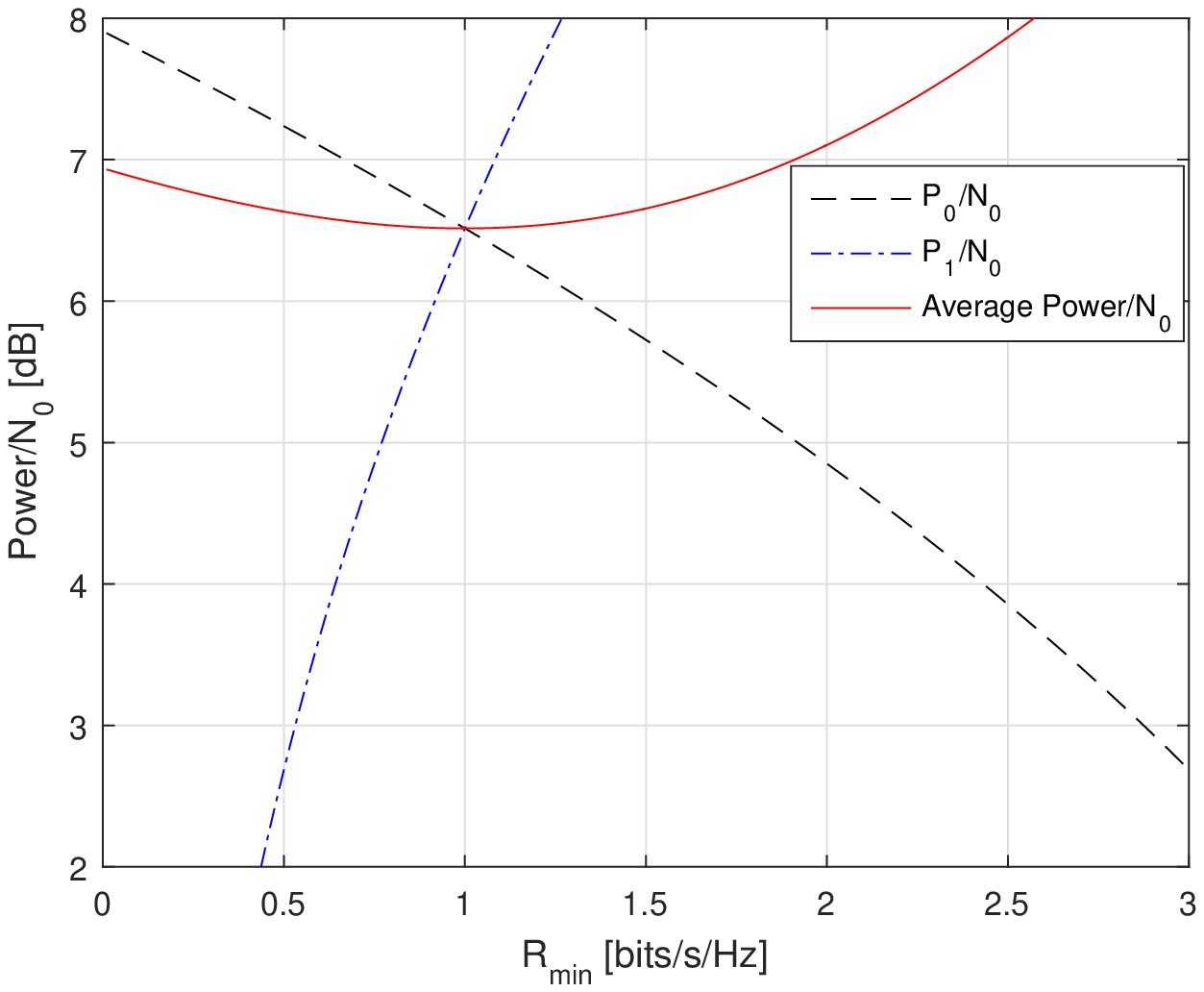}
  \label{fig:2}}
   \caption{System parameters are $R=1,\gamma=0.2,N=1,N_0=1$.}
	\label{fig:general_optimization}
\end{figure*}

The resulting power allocation $P_0$ and $P_1$ are given by
\begin{eqnarray}
P_0&=&\frac{1-2^{\big(\frac{R-R_{min}\pi_1}{\pi_0}\big)}}{\log(1-\frac{\gamma(1-\epsilon_{out})}{1-\gamma})}\\
P_1&=&\frac{1-2^{R_{min}}}{\log(1-\epsilon_{out})}
\end{eqnarray}
The resulting average power $\bar{P}$ in closed form is given by using (\ref{eqn:avg_power}),
\begin{equation}
\Bar{P}=\frac{1-2^{\big(\frac{R-R_{min}\pi_1}{\pi_0}\big)}}{\log\big(1-\frac{\gamma(1-\epsilon_{out})}{1-\gamma}\big)}\pi_0+\frac{1-2^{R_{min}}}{\log(1-\epsilon_{out})}\pi_1
\end{equation}

For this specific case, we evaluate the average power in Fig. \ref{fig:general_optimization}. We fix $\epsilon_1=\epsilon_{out}$ and compute $\epsilon_0$ using (\ref{eqn:close_epsilon_zero}). Then, for the $(\epsilon_{0},\epsilon_1)$ pair, we vary $R_1=R_{min}$ and compute $\bar{P}$ for $R_0$ in (\ref{eqn:close_ratezero}). For a fixed rate $R,\gamma, \epsilon_0$ and $\epsilon_1$, various combinations of $(R_0,R_1)$ provide various average power. We chose values of $\epsilon_1$ such that it is less than $\gamma$ in Fig. \ref{fig:1} and equal to $\gamma$ in Fig. \ref{fig:2}. For every $\epsilon_{out}$, minimum $\bar{P}$ is achieved at a certain $R_{min}$. Please note that the intention is not to optimize $\bar{P}$ in Fig. \ref{fig:general_optimization}; but to show how $\bar{P}$ varies as a function of $R_0$ and there is no mechanism to compute optimal $\bar{P}$ in closed form even for the simplest case of $N=1$. The optimal power allocation can be computed only by jointly searching all possible combinations of $(\epsilon_0,\epsilon_1)$ and $(R_0,R_1)$. We observe that as $\epsilon_{out}$ increases, minimum power is achieved at large $R_{1}$. At $\epsilon_{out}=0.1$, the optimal policy is to transmit with very small $R_1=R_{min}$ while optimal power is achieved at larger $R_1=R_{min}$ when $\epsilon_{out}=0.2$.

Fig. \ref{fig:general_optimization} provides us an interesting insight on optimal power allocation for the problem formulated in (\ref{eqn:optimization}). When $\epsilon_{out}$ and/or target rate $R$ is high, the difference between $R_0$ and $R_1$ is not very large for the optimal power allocation. We use this intuition to propose a relaxed optimization problem in next section and compare performance of solutions of both problems in Section \ref{sect:results}.

\section{Fixed Rate Transmission}
\label{sec:fixed_rate}
In the optimization problem formulated in Section \ref{sect:optimization}, the short term rate guarantee $R_{min}$ over $N$ successive time slots can be quite small as compared to average rate $R$. However, some applications require higher short term minimum rate guarantees such that $R_{min}\to R$. This leads us to a more restrictive optimization problem where packet size for each transmission is fixed to $R$ bits/s/Hz. The resulting optimization problem is formulated as,
\begin{eqnarray}
\label{eqn:optimization:Fixedrate}
&&\min_{\mathbf{P}} \Bar{P}\\
s.t.
&&\begin{cases}
\mathcal{C}_1:\gamma_r\leq \gamma,& 0\leq\gamma\leq 1\\
\mathcal{C}_2:\epsilon_N\leq\epsilon_{out}& 0\leq\epsilon_{out}\leq 1\\
\mathcal{C}_3:R_i=R&\forall i\\
\mathcal{C}_4: P_i\leq P_m,&\forall i,j
\end{cases}
\label{eqn:constraints_fixedrate}
\end{eqnarray}
where $\mathcal{C}_3$ represents the fixed rate constraint and $\mathcal{C}_4$ is the peak power constraint. This implies that largest transmit power at the UE cannot exceed $P_m$ in any state $i$, regardless of the rate. Note that peak power constraint does not explicitly appear in (\ref{eqn:constrains}) because $R_{max}$ depends on $P_m$ via (\ref{eqn:maxrate}) and appears in $\mathcal{C}_4$.
For the modified optimization problem, the objective is to compute power vector $\mathbf{P}$ for fixed rate transmission. The constraints related to packet reception remain the same.
\begin{lemma} For the optimal power allocation in the fixed rate transmission, it holds
 $P_i\leq P_{i+1}$ for all $i \in [0,N]$.
 \label{lem:power_level}
\end{lemma}
\begin{proof}
  It is straight forward to prove by contradiction. If $P_i>P_{i+1}$ and the UE is allowed to enter state $i+1$, an optimal decision is not to transmit in state $i$ at all and wait for a transmission in state $i+1$ which requires less power. This is a birth death process where after every $N-1$ time slots, one transmission is made in state $N$ with power $P_N$. This clearly is suboptimal solution, and makes solving problem for most of the realistic $\gamma$ and $N$ values infeasible.
\end{proof}
\begin{Corollary}
The peak power constraint $P_i\leq P_m, \forall i,j$, reduces to $P_N\leq P_m$.
\label{lem:peak power}
\end{Corollary}
\begin{proof}
From  Lemma \ref{lem:power_level}, $P_i\leq P_{i+1}$, $\forall i$. This implies, $P_N$ is the largest transmit power for any state. Constraining $P_N \leq P_m$ is therefore, enough to apply peak power constraint to full system.
\end{proof}
From Corollary \ref{lem:peak power}, $P_N$ is constrained by $P_m$. However, $P_N$ is also constrained by the power resulting from system parameter $\epsilon_{out}$ via $\mathcal{C}_2$. This implies that the problem is only feasible if the solution satisfies both outage probabilities resulting from the peak power constraint and the outage constraint $\epsilon_{out}$.
Denoting power consumption from $\epsilon_{out}$ by $P_{out}$, the solution is feasible if
\begin{equation}\label{eqn:power_const}
 P_{out}\leq P_N\leq P_m.
\end{equation}
This problem is less flexible as compared to the general optimization problem as no rate adaptation is required at the transmit side. It is worth noting that in spite of reduction in complexity of the problem, the closed form solution of the problem is still not possible due to the challenges explained in Subsection \ref{sec:complex}. In the next section, we discuss special case to get some insight into the problem for the fixed rate transmission case.

\subsection{$N=1$ Case for Fixed Rate Transmission}
\label{sect:N=1}
Let us analyze a special case with $N=1$ for the fixed transmission case.
We compute $\epsilon_0,\epsilon_1$ for a given $\gamma$ for the the fixed transmission rate case, i.e., $R_0=R_1=R$ (as in Section \ref{sect:N=1_general}). It is possible to compute power levels $P_0$ and $P_1$ for the fixed transmission case.

The power levels $P_0$ and $P_1$ are computed from (\ref{eqn:power}) and yield,
\begin{eqnarray}
P_0&=&\frac{1-2^{R}}{\log(1-\frac{\gamma(1-\epsilon_{out})}{1-\gamma})}\\
P_1&=&\frac{1-2^{R}}{\log(1-\epsilon_{out})}
\end{eqnarray}
The resulting average power $\bar{P}$ in closed form is given by using (\ref{eqn:avg_power}),
\begin{eqnarray}\label{}
\Bar{P}&=&\frac{1-2^{R}}{\log\big(1-\frac{\gamma(1-\epsilon_{out})}{1-\gamma}\big)}\pi_0+\frac{1-2^{R}}{\log(1-\epsilon_{out})}\pi_1\\
&=&(1-2^{R})\Bigg(\frac{1}{\log\big(1-\frac{\gamma(1-\epsilon_{out})}{1-\gamma}\big)}\pi_0+\frac{1}{\log(1-\epsilon_{out})}\pi_1\Bigg)\nonumber
\end{eqnarray}
We can classify two distinct regions for the analysis of $\bar{P}$.
\begin{itemize}
  \item {$\epsilon_{out}\leq\gamma$:} For this case, $\epsilon_N=\epsilon_{out}$ and the closed form expressions above hold. When $\epsilon_{out}\leq\gamma$, it implies that outage $\epsilon_N$ in state $N$ should be less than the average outage probability $\gamma$. The optimal decision in the sense of power efficiency in this case is to transmit with power $P_N$ that results in maximum permitted outage $\epsilon_{out}$. This region is termed as bursty packet loss dominant region.
  \item $\epsilon_{out}>\gamma$: For this case, $\epsilon_N=\gamma$. Though, it is permitted to transmit with power $P_N$ that results in $\epsilon_N<\gamma$, but this is not optimal in the sense of minimizing $\bar{P}$. If $\epsilon_{out}>\gamma$, the optimal decision is to transmit with power $P_i,\forall i$ that results in $\epsilon_i=\gamma$, i.e. the power allocation policy is independent of successive packet loss and only determined by average packet loss parameter $\gamma$.
\end{itemize}
From the above characterization, it is clear that $\epsilon_{out}\leq\gamma$ is the most critical region where average power consumption is determined by both the constraints on average packet loss and burst packet loss.
We numerically verify in Section \ref{sect:results} that the power levels computed in closed form for the boundary condition $\epsilon_N=\epsilon_{out}$ are not optimal for $\epsilon_{out}>\gamma$.

As with variable rate problem, the expressions for the power levels cannot be obtained in closed form for $N>1$ when $\epsilon_{out}\leq\gamma$ in spite of fixing rate $R$ for each transmission. The variables $\epsilon_0,\epsilon_1\dots \epsilon_{N}$ are unknown and it is not possible to compute a unique set of $\epsilon_i,\forall i$ in closed form that satisfies $\mathcal{C}_1$ in (\ref{eqn:constraints_fixedrate}).
The optimization problem in (\ref{eqn:optimization:Fixedrate}) is a combinatorial problem as it is hard to compute a unique solution in terms of $\epsilon_i,\forall i$ due to sum of product term in computation of $\gamma_r$. It is therefore, difficult to compute $\mathbf{P}$ that minimizes $\bar{P}$ using convex optimization techniques.

\subsection{Characterization of Critical Regions}
\label{sect:critical_regions}
In the optimization problem in (\ref{eqn:optimization:Fixedrate}), average packet drop rate $\gamma$, successive packet loss constraint $N$ and outage probability $\epsilon_{out}$ affect the average power consumption. It is worth noting that parameter $\gamma$ is the only parameter that controls the 'quantity' of data loss. The parameters $N$ and $\epsilon_{out}$ determine the qualitative effect for the average packet loss rate $\gamma$, i.e., for a fixed $\gamma$, different values of $N$ and $\epsilon_{out}$ result in different QoE for the end user. If we relax $N$ and $\epsilon_{out}$ constraints, we can save more energy at the expense of degradation in QoE without actually dropping more packets.

It is trivial that an increase in the acceptable average packet loss rate $\gamma$ results in a monotonic decrease of average power consumption. However, it is not straightforward to understand the effect of parameters $N$ and $\epsilon_{out}$ on the average power. In \cite{majid_TWC:13}, it has been characterized that there exists a maximum $N$ for a fixed $\gamma$ that results in maximum energy efficiency for the system. Increasing $N$ further, does not result in higher energy efficiency. We further characterize the energy efficiency as a function of qualitative parameters $N, \epsilon_{out}$ by the following lemma:
\begin{lemma}
For a fixed $\gamma$ and $N$, there exists a maximum $\epsilon_{out}=\gamma$ that results in minimum average power consumption. Increasing $\epsilon_{out}>\gamma$ does not help to reduce average power consumption.
\label{lemma:outage}
\end{lemma}
Lemma \ref{lemma:outage} can be proved following the proof of Lemma 1 in \cite{majid_TWC:13}.

We numerically quantify the effect of these parameters on the average system energy consumption in Section \ref{sect:results}. For a fixed $\gamma$, increasing $N$ and/or $\epsilon_{out}$ helps in saving energy in the beginning. This implies that the system is in a region where avoiding successive packet loss has significant effect on average power consumption. An increase in $N$ and/or $\epsilon_{out}$ helps system to drop a fraction $\gamma$ of the packets with more degrees of freedom. When we increase $N$ further, the system enters the region where the gap of $N$ packet drops between two successful packet receptions almost never happens for a given $\gamma$. At this point, it does not matter if the system is allowed to drop more than $N$ successive packets are not. Note that increasing $\epsilon_{out}$ has similar effect as increasing $N$; both permitting packets to be dropped successively within some margins. For small $N$, increasing $\epsilon_{out}$ has significant effect on average energy consumption as compared to large $N$. We provide numerical evidence of this characterization in Section \ref{sect:results}.

\section{Stochastic Optimization}
\label{sect:stochastic}
The combinatorial optimization problems in (\ref{eqn:optimization}) and (\ref{eqn:optimization:Fixedrate}), which are not solvable with regular optimization techniques, can approximately be solved using stochastic optimization methods. There are a few heuristic techniques in literature to solve such problems like genetic algorithm, Q-learning, neural networks, etc. All of these techniques rely on randomized inputs to compute a solution at reduced computational complexity as compared to exhaustive search.
Simulated Annealing (SA) is another similar stochastic optimization algorithm with the distinct feature that it helps avoid the solution to get stuck in local minima by introducing a probabilistic process called 'muting' as explained later in this section. The algorithm was originally introduced in statistical mechanics, and has been applied successfully to networking problems \cite{majid_TWC:13, majid:sys2016}. Based on its ability to compute global minima with high probability, we use SA algorithm to solve optimization problem in (\ref{eqn:optimization}) and (\ref{eqn:optimization:Fixedrate}).

\renewcommand{\baselinestretch}{1}
\begin{algorithm}
\caption{Optimization by SA Algorithm for the General Case}
\KwIn{$(\mathbf{A}, T_m,\gamma,\epsilon_{out},P_m$)\;}
$T_m$ = lower bound on temperature\;
$P_{a,0}$= Compute $\bar{P}$ as a function of initial $\mathbf{A}$\;
$\bar{P}^* =P_{a,0}$; $\mathbf{A^*} = \mathbf{A}$\;
$T_b=T_0$\;
\While{$T_b \geq T_m$}{
$T_b = \frac{T_0}{c_{\rm sa}\cdot b+1}$\;
\For{i=0 \KwTo n}{
 Generate a random $\mathbf{\hat{A}}$\;
 Compute $\gamma_r$ for $\mathbf{\hat{A}}$\;
 Evaluate $\mathcal{C}_2$ and $\mathcal{C}_3$\;
 \If {$\mathcal{C}_2$ and $\mathcal{C}_3$ satisfied}{
 \For{j=0 \KwTo r}{
 Generate $\textbf{R}$ and evaluate $\mathcal{C}_1$ and $\mathcal{C}_4$\;
 \If{$\mathcal{C}_1$ and $\mathcal{C}_4$ are satisfied}{
 Solution feasible\;
  Compute power vector $\hat{\mathbf{P}}$ as a function of $\mathbf{\hat{A}}$ using (\ref{eqn:power})\;
  Compute average power $\hat{P}_a$ in (\ref{eqn:avg_power})\;
  $s$ = A random number in range $[0,1]$\;
    \If {$s<\exp \big(\frac{-(\hat{\bar{P}}- \bar{P}_a)}{T}\big)$}{
        $\bar{P_{a}}=\hat{P}_a$\;
        \If {($\hat{\bar{P}}\leq \bar{P}^*$)}{
        $\bar{P}^* =\hat{\bar{P}}$;
        }
    }
  }
   \Else{
   Solution Infeasible\;
   }
 }
}
  \Else{
   Solution Infeasible\;
   }

}
}
\KwOut{$(\bar{P}^*,\mathbf{A}^*$);}
\label{algorithm_1}
\end{algorithm}

In SA algorithm, a random configuration in terms of transition probability matrix $\mathbf{A}$ is generated in each iteration. Average power $\bar{P}$ is evaluated only if constraints in (\ref{eqn:constrains}) are met. If the evaluated $\bar{P}$ is less than the previously computed best solution, the candidate set of outage probabilities $\epsilon_i$, $\forall i$ are selected as the best available solution. However, the candidate set $\epsilon_i$, $\forall i$ can be treated as the best solution with a certain temperature dependent probability even if the new solution is worse than the best known solution. This step is called \emph{muting} and helps the system to avoid local minima. The muting occurs frequently at the start of the process as the selected temperature is very high and decrease as temperature is decreased gradually, where temperature denotes a numerical value that controls the muting process.

In literature, various cooling temperature schedules have been employed according to the problem requirements, such as Boltzmann annealing, fast annealing and adaptive cooling. The cooling schedule determines the convergence rate of the solution. If temperature cools down at a fast rate, the optimal solution can be missed. On the other hand, if it cools down too slowly, optimization requires large amount of time. In this work, we employ fast annealing (FA) \cite{FA} because it provides us reasonably good results. In FA, it is sufficient to decrease the temperature linearly in each step $b$ such that,
\begin{equation}
\label{eqn:BA} T_b = \frac{T_0}{c_{\rm sa}\cdot b+1}
\end{equation}
where $T_0$ is a suitable starting temperature and $c_{\rm sa}$ is a constant, which depends on the requirements of the problem. After a fixed number of temperature iterations, when muting fully stops, the best solution is accepted as an approximation to the optimal solution. Note that the solution provided after a fixed number of temperature iterations is used to keep the computational complexity manageable. To show the convergence behaviour of the solution provided by SA, we compare the SA approximated results with the analytical results for the $N=1$ case.\footnote{A lot of literature is available on providing more accurate measures of convergence for SA algorithm \cite{Granville_TPA:1994}, but going in rigorous mathematical details on the convergence of the approximated solution is out of scope of this paper.}

To apply the SA algorithm and solve the optimization problem in (\ref{eqn:optimization}), we use the following 2-step process:
\begin{enumerate}
  \item First generate a random set of $\epsilon_i\forall i$ and evaluate if $\mathcal{C}_2$ and $\mathcal{C}_3$ are met. The candidate solutions which do not meet $\mathcal{C}_2$ and $\mathcal{C}_3$, are not feasible solutions and they are dropped.
  \item For the candidate solutions that meet $\mathcal{C}_2$ and $\mathcal{C}_3$, we solve the following programming problem:\\
      Find $\mathbf{R}$ that meets $\mathcal{C}_1$ and $\mathcal{C}_4$. That constitutes a feasible solution. For all feasible solutions, we evaluate $\bar{P}$ and choose the $\mathbf{P,R}$ vectors that minimize $\bar{P}$. Note that we need to generate randomized vector $\mathbf{R}$ between values $R_{min}$ and $R_{max}$ to generate a candidate solution.
      \end{enumerate}
Pseudocode for the optimization of problem using SA is presented in Algorithm \ref{algorithm_1}. The complexity
of the solution depends on parameter $N$. For large $N$, size of transition probability matrix grows and it becomes computationally expensive to calculate the optimal matrix.

\begin{algorithm}[t]
\caption{Optimization by SA Algorithm For Fixed Rate Transmission}
$T_b = \frac{T_0}{c_{\rm sa}\cdot b+1}$\;
\For{i=0 \KwTo n}{
 Generate a random $\mathbf{\hat{A}}$ and compute $P_N$\;
 \If {$(\max(\mathbf{P})\leq P_m)\AND (P_N\geq P_{out})$}{
 Solution feasible\;
 Compute $\gamma_r$ for $\mathbf{\hat{A}}$\;
  \If{$\gamma_r<\gamma$}{
  Compute power vector $\hat{\mat{P}}$ as a function of $\mathbf{\hat{A}}$ using (\ref{eqn:power})\;
  Compute average power $\hat{P}_a$ in (\ref{eqn:avg_power})\;
  $s$ = A random number in range $[0,1]$\;
  \If {$s<\exp \big(\frac{-(\hat{\bar{P}}- \bar{P}_a)}{T}\big)$}{
        $\bar{P_{a}}=\hat{P}_a$\;
     \If {($\hat{\bar{P}}\leq \bar{P}^*$)}{
        $\bar{P}^* =\hat{\bar{P}}$;
        }
     }
   }
   }
   \Else{
   Solution not feasible\;
  }
 }
\KwOut{$(\bar{P}^*,\mat{A^*}$);}
\label{algorithm_2}
\end{algorithm}


The solution for the fixed rate transmission optimization problem in (\ref{eqn:optimization:Fixedrate}) is computed using the SA algorithm in a similar way, but it is relatively less complex. As rate is fixed for every transmission, a one step feasible solution comprising $\epsilon_i\forall i$ is selected from the randomly generated candidate solutions that meets $\mathcal{C}_1-\mathcal{C}_4$. For every feasible solution, the objective function is evaluated and the solution that minimizes $\bar{P}$ is selected. To minimize the repetition, pseudocode for one temperature iteration is presented in Algorithm \ref{algorithm_2}.


\section{Numerical Results}
\label{sect:results}

We perform a numerical evaluation of the proposed scheduling scheme in this section. We consider a Rayleigh fading channel with mean $1$ for the point to point link. Peak power is set to $20$ dBW for all numerical examples while the noise variance $N_0$ equals one. $R_{min}$ is set to a small value of 0.001 bits/s/Hz to allow system to choose almost any $R_0\leq R_m,R_1\leq R_m$ combination that minimizes $\bar{P}$. The cooling schedule from (\ref{eqn:BA}) is applied in SA algorithm where number of iterations per temperature value is fixed.

We study the effect of packet loss parameters on the average power consumption for the special case $N=1$ in Fig. \ref{fig:n1_case}, where the results are evaluated using both closed form expressions derived in Sections \ref{sect:N=1_general} and \ref{sect:N=1}; and the SA framework developed in Section \ref{sect:stochastic}. Average transmit power is plotted for $N=1$ and $\gamma=0.2$ in Fig. \ref{fig:n1_case}. Note that $\epsilon_{out}=\epsilon_N$ in the closed form expression. For the fixed power case, the average power consumption is a convex function in $\epsilon_{out}$ for a fixed $\gamma$ and $N$, and a unique optimal $\epsilon_{out}$ can be identified. Let us call it $\epsilon_{out}^*$. If system parameter $\epsilon_{out}\leq\epsilon_{out}^*$, it results in high average power. However, if $\epsilon_{out}>\epsilon_{out}^*$, the system has more flexibility and it is optimal to set $\epsilon_N=\epsilon_{out}^*$ instead to save power. The optimized results with the SA method match closely with the analytically computed results for $\epsilon_{out}\leq\epsilon_{out}^*$ which validate the accuracy of the solution provided by the SA algorithm. For $\epsilon_{out}>\epsilon_{out}^*$, SA method provides the optimal solution in contrast to the suboptimal solution where $\epsilon_N = \epsilon_{out}$ is enforced. For the case of variable rates power allocation, the analytical results and the solutions from the SA method match closely. However, the difference is more in the case of fixed rate. Note that the results for the analytical case cannot be fully computed in closed form and the optimal rate is computed using a scalar search for the optimal $R_{min}$ at fixed $\epsilon_{0}$ values which introduces numerical inaccuracy for the 'semi-analytical' solution for $N=1$.

\begin{figure}
\center
\includegraphics[width=3.5in]{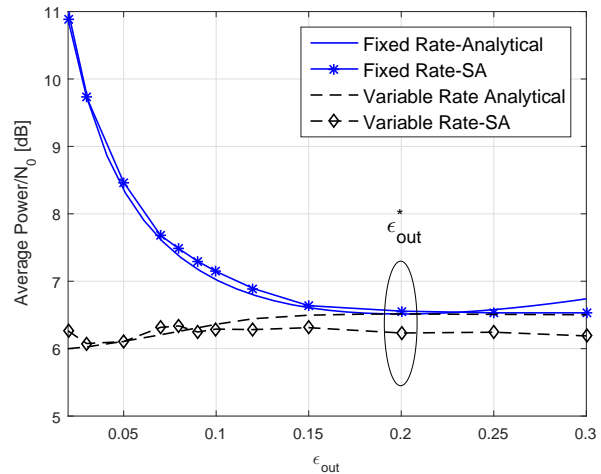}
\caption{Average power for both variable and fixed rate cases for the special case $N=1$. Analytically computed solutions are compared with the solutions produced using the SA algorithm. The target rate $R$ is set to 1 bits/s/Hz and $\gamma$ is fixed to 0.2.}
\label{fig:n1_case}
\end{figure}

\begin{figure}
\center
\includegraphics[width=3.5in]{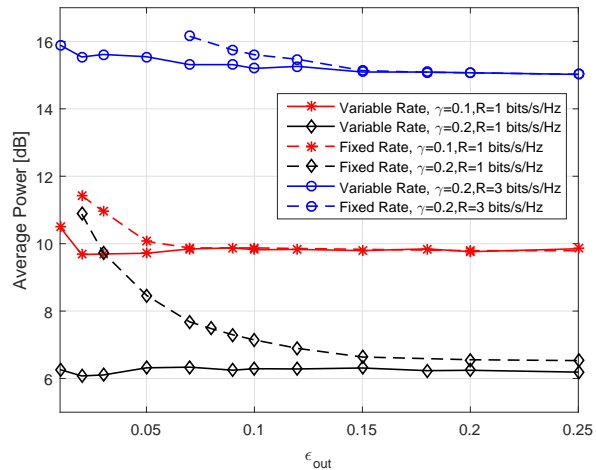}
\caption{Average power consumption for variable and fixed rate schemes for $N=1$ and the SA algorithm.}
\label{fig:scheme_comp}
\end{figure}

\begin{figure*}
\centering
  \subfigure[Fixed Rate, R=1 bits/s/Hz]
  	{\includegraphics[width=3.0in]{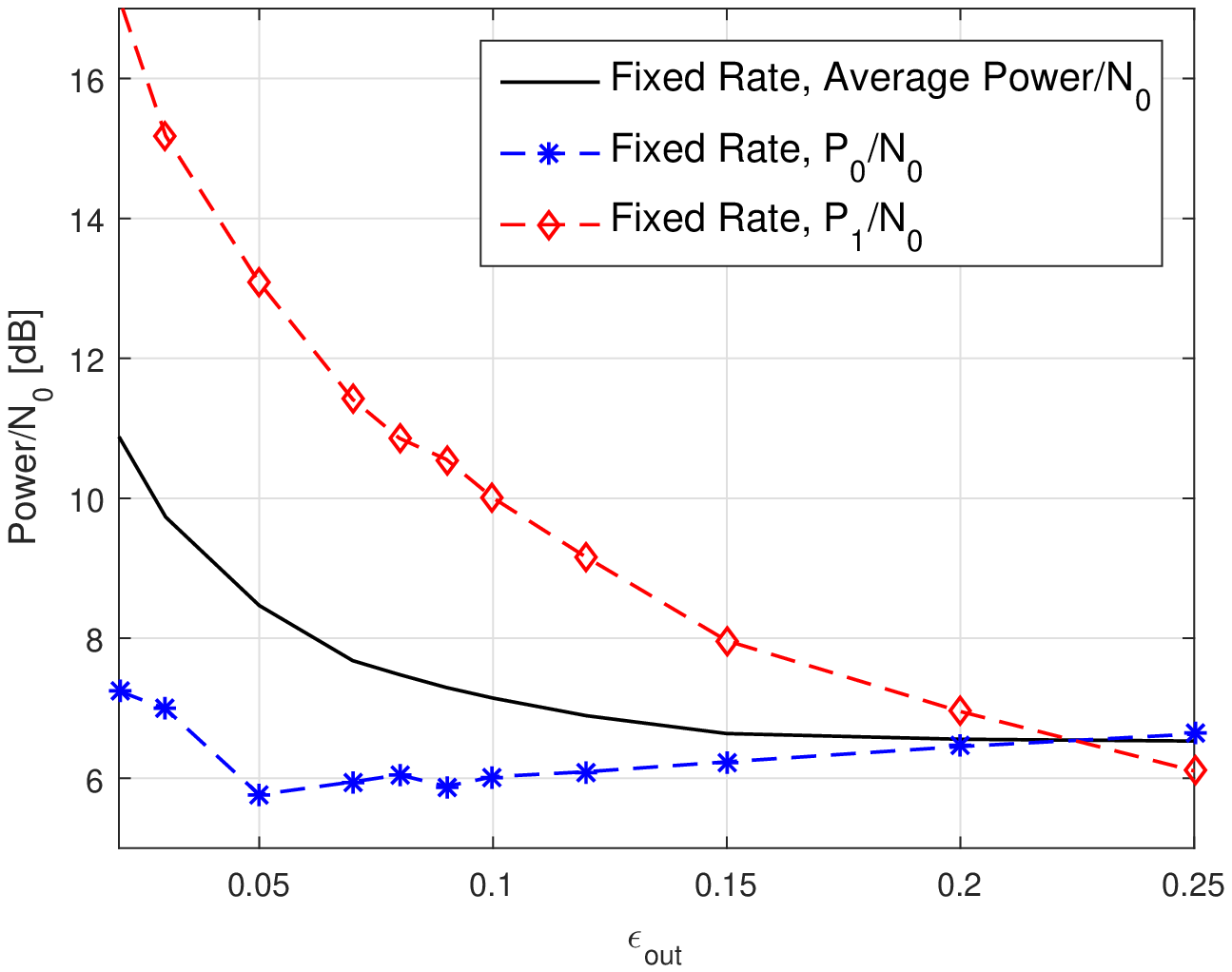}
  \label{fig:fixpower}}
 \subfigure[Variable Rate, R=1 bits/s/Hz, $R_{min}=0.001$ bits/s/Hz]
  	{\includegraphics[width=3.0in]{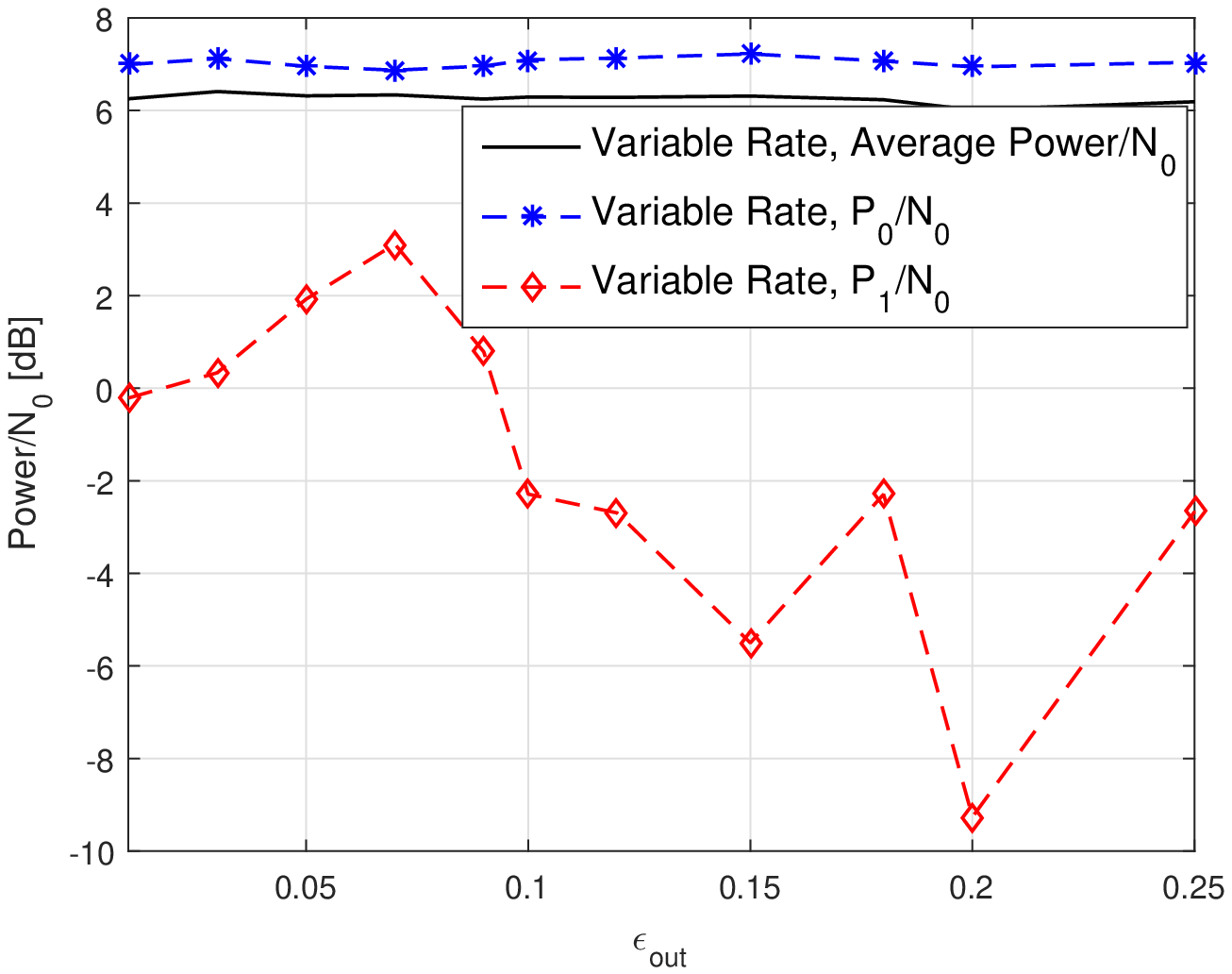}
  \label{fig:varpower R1}}
  \subfigure[Variable Rate, R=1 bits/s/Hz, $R_{min}=0.5$ bits/s/Hz]
  	{\includegraphics[width=3.0in]{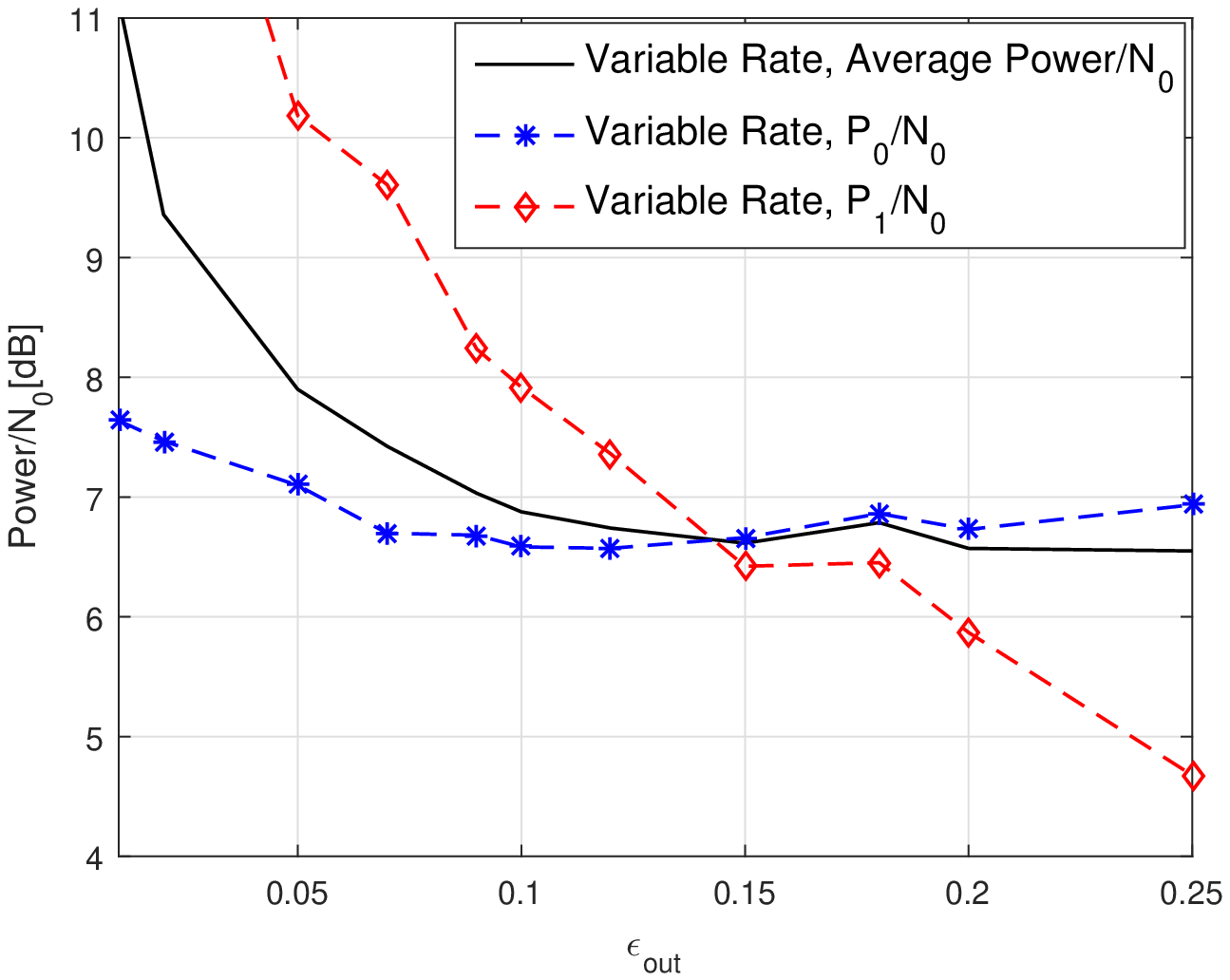}
  \label{fig:varpower R1_minrate}}
    \subfigure[Variable Rate, R=3 bits/s/Hz, $R_{min}=0.001$ bits/s/Hz]
  	{\includegraphics[width=3.0in]{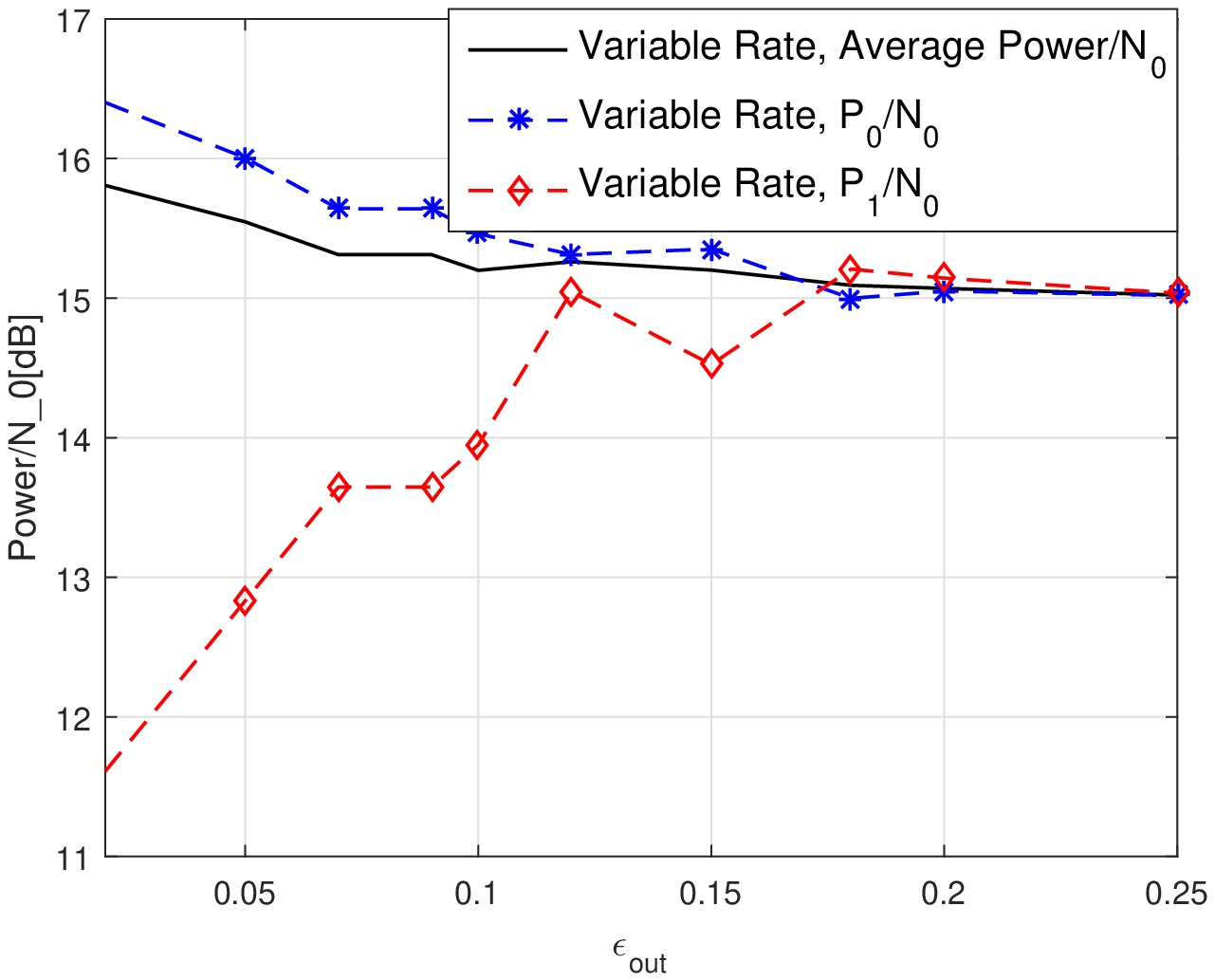}
  \label{fig:varpower R3}}
   \caption{Power levels for the fixed and variable rate scheme for the individual states. The system parameters are $N=1, \gamma=0.2$.}
	\label{fig:power levels}
\end{figure*}

Fig. \ref{fig:scheme_comp} compares the average power for both schemes for $N=1$ case. We use the SA method to compute the results for all the examples in rest of this section.\footnote{It should be noticed that the results obtained from the SA algorithm always show some irregular points due to inherent randomness of the heuristic algorithm in computing the solution.} The results clearly show that variable rate scheme performs better than the fixed scheme at small $\epsilon_{out}$. This is attributed to more flexibility in choosing rates for different states. Another important observation can be noted for $R=3$ bits/s/Hz. Fixed rate scheme cannot provide any quality of service for $\epsilon_{out}<0.07$ and $\gamma=0.2$. The flexibility of variable rate scheme allows to achieve almost identical average power for all $\epsilon_{out}$ including the smaller ones. This leads us to the conclusion that rate adaptation is useful at small $\epsilon_{out}$, while fixed rate transmission becomes almost as efficient at high $\epsilon_{out}$ at reduced complexity. This is clearly more evident at higher rates ($R=3$ bits/s/Hz).

In Fig. \ref{fig:power levels}, we investigate the power allocation for individual states for both schemes. Interestingly, power allocation is opposite in both schemes at small $\epsilon_{out}$. For the fixed rate scheme, the optimal power allocation requires to transmit with small power in state '0' and large power in state '1'. As $R=1$ in both states, it implies $\epsilon_0\leq\epsilon_1\leq \epsilon_{out}$ for the optimal power allocation. For the variable rate case when $R_{min}$ is small (0.001 bits/s/Hz) as in Fig. \ref{fig:varpower R1} and Fig. \ref{fig:varpower R3}, the optimal power allocation requires $P_0\geq P_1$. The optimal policy is to transmit with higher rate and power $P_0\leq P_{m}$ in state '0' and with substantially small power and rate in state 1. This results in optimal average power when rate adaptation is allowed. In Fig. \ref{fig:varpower R1_minrate}, when $R_{min}$ is increased to 0.5 bits/s/Hz for $R=1$, the power allocation for variable rate scheme resembles more to fixed rate scheme due to decrease in flexibility in rate allocation. The fixed transmission policy suffers from the constraint $R_0=R$ and any dropped packets have to be compensated in state $1$ with more power. When $\epsilon_{out}\to \epsilon_{out}^*$ for variable rate power allocation, $P_1$ is relatively small as compared to $P_0$ if $R$ is small (Fig. \ref{fig:varpower R1}) and peak power constraint is large. If $R$ is large as in Fig. \ref{fig:varpower R3}, $P_1$ is increased to meet the rate constraints. However, the intention of the system is to exploit state '0' by adapting power and rate as there is no constraint on $\epsilon_0$.

\begin{figure}
\center
\includegraphics[width=3.5in]{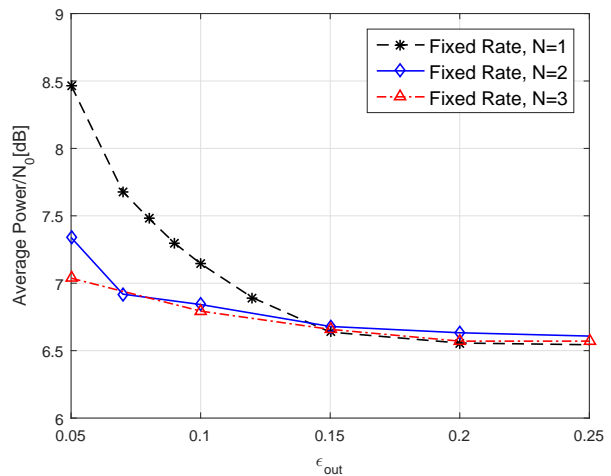}
\caption{Average power as a function of packet loss parameters for different $N$ for the fixed rate scheme. $\gamma$ is fixed to 0.2 and $R=1$ bits/s/Hz.}
\label{fig:general_n}
\end{figure}

Fig. \ref{fig:general_n} compares the average power consumption for the fixed rate scheme for the case $N=1,2,3$ and $\gamma=0.2$. The power levels are optimized using SA algorithm. It is evident that the resulting average power converges for all $N$ to the same minimum value at $\epsilon_{out}^*$. When $\epsilon_{out}\leq\epsilon_{out}^*$, an increase in $N$ for a fixed $\gamma$ helps to reduce average power consumption in general (specially at small $\epsilon_{out}$). More flexibility in packet dropping parameters provides more degrees of freedom and results in energy savings. When $\epsilon_{out}>\epsilon_{out}^*$, the effect of large $N$ vanishes and power saving depends solely on average packet dropping parameter.

\section{Conclusion}
\label{sect:conclusions}
We consider energy efficient scheduling and power allocation for the loss tolerant IoT applications. Data loss is characterized as a function of average and successive packet loss, and the probability that the successive packet loss constraint is not met. These parameters jointly define the QoE and context for an IoT application. In contrast to average packet loss parameter, other loss parameters depend on the packet loss order without actually changing the number of lost packets. By considering bursty packet loss a form of contextual information, we provide another degree of freedom in the scheduling algorithm which can be exploited to reduce energy consumption. Without CSIT, we formulate the average power optimization problem as a function of data loss parameters. First, the generalized power optimization problem is discussed where the transmitted packets are adapted in size such that an average rate and minimum packet size guarantee is provided. Then, we relax the problem to the case where the packet size is fixed for all transmissions. Both of the optimization problems are combinatorial in nature and require a stochastic optimization technique to solve them. For both problems, we compute analytical expressions of average power as a function of system parameters for the special case $N=1$ and compare it with the solution obtained from the proposed simulated annealing algorithm. Both of the analytical results match quite well and validate the solution provided by the heuristic simulated annealing algorithm. We numerically study performance of both schemes and show dependency of power consumption on the parameters that depend on the order of packet drop in addition to packet drop rate.

\bibliographystyle{IEEEtran}
\bibliography{bibliography}

\end{document}